\documentclass[lettersize,journal]{IEEEtran}
\usepackage{amsmath,amsfonts}
\usepackage{algorithmic}
\usepackage{array}
\usepackage[caption=false,font=normalsize,labelfont=sf,textfont=sf]{subfig}
\usepackage{textcomp}
\usepackage{stfloats}
\usepackage{url}
\usepackage{xcolor}
\usepackage{verbatim}
\usepackage[shortlabels]{enumitem}
\usepackage{amssymb}
\usepackage{bbm}
\usepackage{subfig}
\usepackage{cite}
\usepackage{float}
\usepackage{multicol}
\usepackage{epsfig}
\usepackage{amsthm}
\usepackage{cleveref}
\usepackage{wasysym}
\usepackage{ragged2e}
\usepackage{caption}
\DeclareMathOperator*{\argmin}{argmin}
\newtheorem{theorem}{Theorem}[]
\newtheorem{definition}{Definition}[]

\newtheorem{lemma}{Lemma}[]
\newtheorem{remark}{Remark}
\newtheorem{case}{Case}
\newtheorem{subcase}{Subcase}
\newtheorem{Proposition}{Proposition}[]
\newtheorem{Claim}{Claim}[]
\usepackage{graphicx}
\usepackage{balance}
\begin{document}
\title{Towards Maximizing Nonlinear Delay Sensitive Rewards in Queuing Systems}

\author{%
  \IEEEauthorblockN{Sushmitha Shree S\textsuperscript{*}\textsuperscript{{\textsection}}, Avijit Mandal\textsuperscript{{\textsection}}, Avhishek Chatterjee, Krishna Jagannathan}\\
  \IEEEauthorblockA{Department of Electrical Engineering, Indian Institute of Technology Madras, Chennai 600036, India\\
                   \{sushmithasriram, avijitbesu1995\}@gmail.com, 
                   \{avhishek, krishnaj\}@ee.iitm.ac.in 
                    }
                    }

\maketitle
\begingroup\renewcommand\thefootnote{*}
\footnotetext{Corresponding author}
\endgroup

\begingroup\renewcommand\thefootnote{\textsection}
\footnotetext{Equal technical contribution}
\endgroup

\pagenumbering{gobble}
\begin{abstract}

We consider maximizing the long-term average reward in a single server queue, where the reward obtained for a job is a non-increasing function of its sojourn time. The motivation behind this work comes from multiple applications, including quantum information processing and multimedia streaming. We introduce a new service discipline, shortest predicted sojourn time (SPST), which, in simulations, performs better than well-known disciplines. We also present some limited analytical guarantees for this highly intricate problem.
\end{abstract}

\section{Introduction}
\IEEEPARstart{J}{ob} scheduling in single server systems is one of the most widely researched areas due to its diverse applications \cite{Potts2009}. Historically, the design of service disciplines focused on optimizing the average linear functions of sojourn times (a.k.a response times). Under this performance measure, the discipline that processes the job with the shortest remaining processing time (SRPT) proves to be optimal \cite{Miller1966}. However, almost no work considers optimizing nonlinear functions of sojourn times, which have become crucial in many emerging applications, a few of which we briefly discuss. 
\begin{enumerate}
    \item Quantum information processing: The quantum bits (qubits) that are generated for sequential processing by a circuit or for transmission over a channel undergo decoherence while waiting to be processed or transmitted \cite{Nielsen2010}. The effective information extracted out of a stream of bits is the stationary expectation of a non-increasing function of the sojourn time \cite{krishna2019}.

    \item Multimedia streaming: In streaming applications, delayed packets cause stream to break or pause. Hence, the value of a multimedia packet decreases with its delay \cite{Kurose2012}. 
    
    \item Delay sensitive online services: In online service platforms like ride-sharing and food delivery, customers' satisfaction and hence, in turn, ratings often depend on the delay in the service. In fact, in many settings, user dissatisfaction due to delays cannot be compensated by better service or other promotional offers \cite{Daugherty2010,Dewan1990}.
\end{enumerate}

Although optimizing nonlinear functions of sojourn times is crucial for these applications, there is hardly any study aimed at optimizing the average nonlinear functions of sojourn times, even in a single server case. This paper takes a few steps towards this goal, and is motivated by the aforementioned applications.

\subsection{Related work and Motivation}
In work conserving single server queuing systems, jobs can arrive arbitrarily. When the service requirements (job sizes) are known, SRPT minimizes the average sojourn time regardless of the arrival and service distributions \cite{Miller1966}. Under SRPT, the job in service has the least remaining processing time, and an incoming job preempts the server only if its processing time is shorter than the remaining processing time of the job in service. Specifically, SRPT minimizes the sojourn time for every arrival sequence \cite{Miller1966}. In other words, SRPT is said to be sample-path optimal. Schrage in \cite{Schrage1968} first discussed the proof of optimality of SRPT, followed by Smith in \cite{Smith1978}. SRPT gained popularity thereon that prompted the analysis of its performance guarantees \cite[Chapter~33]{Balter2013}, the evaluation of its fairness among jobs \cite{Balter2001}, its implementation in web servers \cite{Balter2000} and its extension to multiple server systems \cite{Balter2018, Vaze2020}. 

Unlike classical queuing systems that assume no constraints on the waiting times, jobs do come with fixed deadlines in certain applications \cite{Barrer1957}. If the server does not process a job within its deadline, it drops off the queue and never returns for service (balking or reneging). The dynamics of these systems have been extensively investigated under multiple settings \cite{Daley1965, Baccelli1984, Kargahi2004, Pascal2013, Ahmadi2021, Balter2015}. The most common performance measure here is the overall loss fraction that captures the fraction of jobs lost out of the total arrivals to the system. The earliest deadline first (EDF) discipline is shown to be optimal in minimizing the overall loss fraction irrespective of the service requirements \cite{Liu1973, Towsley1988}. However, minimizing the overall loss fraction does not always guarantee the minimum average sojourn time. Therefore, it is reasonable to associate a reward for each job that captures the trade-off between the fraction of loss and the average sojourn time in the system. In \cite{Raviv2018}, the deadline and reward of jobs are known upon arrival, and the optimal policy that maximizes the rewards per service requirement of served jobs has been studied. \cite{Hyon2020} and \cite{Yu2018} present a similar line of work. Nevertheless, in real-time systems, neither the deadlines nor the rewards of jobs are known to the server.

Our work is inspired by the applications such as quantum information processing and multimedia streaming. In these applications, the information in the jobs (qubits in quantum systems \cite{krishna2019} and data packets in multimedia systems \cite{Draper2005}) become useless or erased after a certain deadline. Unlike impatient customers, the jobs do not drop off the queue; however, processing them after their deadline may not be useful to the system. 

For instance, in the quantum setting, qubits arrive sequentially at a quantum system and wait in the queue until they are processed. While a qubit waits in the queue, it undergoes decoherence due to its interaction with the environment \cite{Nielsen2010}. The decoherence of a qubit leads to the erasure of its information, and the probability of qubit erasure is modeled as an explicit function of its sojourn time. For example, if a qubit waits for $W$ units of time in the system, then the probability of its erasure is modeled as $p(W)=1-\exp{(-\kappa W)}$ for some $\kappa>0$, where $\kappa$ is the characteristic parameter of the quantum system \cite{krishna2019}. In other words, a qubit with sojourn time $W$ is associated with a reward of the form $\exp{(-\kappa W)}$ for some $\kappa>0$. A similar model is relevant in the areas of multimedia streaming \cite{Draper2005} and crowdsourcing \cite{Avhishek2017}. 

The information capacity of quantum erasure channels has been derived irrespective of the service discipline in \cite{krishna2019}. Specifically, this capacity is proportional to $\mathbb{E}[\exp{(-\kappa W)}]$, where the expectation is over the limiting distribution of the sojourn times. The goal of maximizing the capacity of quantum erasure channels poses an interesting problem and reduces to maximizing the average nonlinear function of sojourn times (rewards). Our work is inspired by such a setting. In particular, this work aims to maximize the average nonlinear functions of the form $\exp{(-\kappa W)}$ for some $\kappa>0$ from a scheduling perspective.

\subsection{Contributions}
In this work, we consider a work conserving single server queuing system in which the service requirements of the jobs are known upon arrival. Each job is associated with a reward based on its sojourn time. Specifically, the reward of a job is a specified non-increasing function, possibly nonlinear in its sojourn time. This work aims to identify the service discipline that maximizes the long-term average of rewards. Since the rewards are a function of sojourn times, this essentially ensures the maximization of the long-term average of rewards while processing the maximal number of jobs.

We view this problem for two arrival models. Firstly, we consider batch arrival models in which an arbitrary number of jobs arrive at the server at the same instant. In this model, we show that processing the jobs with the shortest service requirements maximizes the long-term average rewards of the system. In addition, we show that this result holds for all monotonic functions of sojourn times.

Next, we analyze a more realistic arrival model in which jobs arrive according to a stochastic process. It is well-known that SRPT maximizes linear rewards \cite{Miller1966} for all arrival sequences and service distributions; however, it is unclear if SRPT maximizes nonlinear rewards. For a single server system with a unit service rate, simulations show that SRPT does not perform better for some arrival and job size distributions. Indeed, we find that identifying a discipline that maximizes any monotonic function of sojourn times poses a difficult problem. This is mainly because the performance of the service disciplines has a complex dependence on the i) arrival and service distributions, ii) job sizes, and iii) function of sojourn times. Certainly, the simulation of the performance of existing disciplines shows that there is no clear winner for all arrival sequences and functions of sojourn times. To reduce the complex dependency on the function of sojourn times, we focus only on rewards of the form $\exp{(-\kappa W)}$ for some $\kappa>0$, where $W$ represents the sojourn time. These functions have practical implications in applications such as quantum information systems and multimedia streaming, as mentioned before.

In this work, we introduce a service discipline, \emph{shortest predicted sojourn time (SPST)} and analyze its performance in this setting. According to SPST, a job in service has the least predicted sojourn time. Through simulations, we infer that the performance of SPST is promising for all arrival and job size distributions. However, analytically proving this for all arrival distributions and job sizes is still a hard problem. Therefore, we assume a simple model where jobs of the same size arrive at the server with stochastic interarrival times. Due to the combinatorial intricacies, we compare the performance of SPST with only the first come first serve (FCFS) discipline for this model. In particular, we show that the long-term average reward under SPST is higher than that under FCFS for $\kappa\ge \log_e{2}$. Moreover, it is evident from this result that there is no optimal service discipline that maximizes the long-term average of rewards of the form $\exp{(-\kappa W)}$ for all $\kappa$. 

\subsection{Organization}
The rest of the paper is organized as follows: Section \ref{system_model} gives an overview of the system with batch arrivals and stochastic arrivals. Section \ref{batch_arrivals} and \ref{Section:Sequential} discuss the main results for these two scenarios respectively. Under stochastic arrivals, the simulations of the performance of SPST and other disciplines are discussed in section \ref{simulation_results}. Followed by the analytical findings of the performance comparison of SPST with FCFS that are covered in section \ref{analytical_results}. Proofs are detailed in the Appendix.

\section{System Model}\label{system_model}
We consider a discrete-time work conserving single server queue with unit service rate. The jobs with integer sizes $\{S_i, i\in \mathbb{N}\}$ arrive randomly at the server. These jobs are indexed by positive integers according to their arrivals, with the ties broken arbitrarily. At the beginning of every time slot, the server can change its service from one job to another based on the service discipline. Each job waits in the queue before being served, and the total time it spends in the system is known as its sojourn time. For a job indexed by $i$, $W_i$ represents its sojourn time, and $f(W_i)$ is the associated reward, where $f(\cdot)$ is a non-increasing function. This work aims to find a service discipline that maximizes the long-term average of rewards.

In this work, we consider two scenarios: (i) batch arrivals and arbitrary job sizes and (ii) stochastic arrivals and stochastic job sizes. In the first scenario, as the name suggests, $n$ jobs arrive at time $0$ and their sizes are $\{S_i:1\le i \le n\}$. In this context, our goal is to find  a service discipline that maximizes the accumulated reward, $\sum_{i=1}^n f(W_i)$, for any positive integer $n$ and $\{S_i:1\le i \le n\}$. 

In the second setting, jobs arrive according to some point process with i.i.d. positive inter-arrival times $Y_1, Y_2, \ldots$.  Job sizes $\{S_i\}$ are also i.i.d. positive random variables. In this scenario, the goal is to find a {\em stationary} service discipline $\pi$ under which the long-term average reward, $\lim \limits_{n\to\infty}\frac{1}{n}\sum_{i=1}^n f(W_i)_\pi$, is maximum. Note that whenever $\mathbb{E}[Y_1]>\mathbb{E}[S_1]$, this limit exists almost surely for any work conserving stationary service discipline.  

\section{Batch Arrivals and Arbitrary Job Sizes}
\label{batch_arrivals}
In any queue setup, jobs are assumed to arrive singly at a server. However, this is not the case in all real-world scenarios. Jobs do come in batches of fixed or random sizes \cite{chaudhry1983} as in the case of cloud-based data processing. This section characterizes the service discipline that maximizes the accumulated reward in a queuing system with a single batch of job arrivals.

\begin{definition}[{Shortest job first (SJF)\cite[Chapter~31]{Balter2013}}] Under this non-preemptive service discipline, whenever the server frees up, it serves the job with the shortest service requirement to completion. \label{definition_SJF}
That is, at any time $t$, the index of the job in service is $k=\argmin\limits_i S_i$. Ties are broken arbitrarily.

\end{definition}
In the case of a single batch of arrivals, the jobs are served in increasing order of their sizes under SJF.  

\begin{theorem}\label{theorem:batch}
    For any batch size $n$ and any service requirements $\{S_i:1\le i\le n\}$, SJF maximizes $\sum_{i=1}^n f(W_i)$.

\end{theorem}
The proof of theorem \ref{theorem:batch} is a direct consequence of the following lemma.
Consider that a bunch of $n$ jobs arrive at an arbitrary time $t$. Let $\{S_k, k \in [1, n]\}$ denote their sizes and $\{J_i, i \in [1, n]\}$ be the job labels in increasing order of their sizes i.e., if, for $J_i$, $J_j$ such that $i<j$, then $ S_i \leq S_j \; \forall i, j \in [1,n]$. Let $A_1$ be the service discipline that serves the jobs in the order  $\{J_1, J_2, \ldots, J_k, J_{k+1}, \ldots, J_n\}$. Consider another discipline $A_2$ with order of service $\{J_1, J_2, \ldots, J_{k+1}, J_k, \ldots, J_n\}$. Let $R_\pi$ denote the accumulated reward under service discipline $\pi$. Here, $R_{\pi}=\sum\limits_{i=1}^n f(W_i)_\pi$.

\begin{lemma}{\label{proposition:Batch}}
For a non-increasing function $f$, $R_{A_1} \ge R_{A_2}$.
\end{lemma}
\begin{proof}[Proof of lemma \ref{proposition:Batch}]
    In a work-conserving system with order of service $\{l_i, i \in [1,n]\}$, the sojourn time of job at index $l_k$, $W_{l_k}=W_{l_{k-1}}+S_{l_k}$. Equivalently,  $W_{l_k}=\sum\limits_{i=1}^{k} S_{l_i}$. Clearly, 
    \begin{align}
        f(W_{l_i})_{A_1}=f(W_{l_i})_{A_2} \quad \forall i \neq k, k+1. \label{batch:1}
    \end{align} 
    So, it is sufficient to compare $f(W_{l_k})+f(W_{l_{k+1}})$ under $A_1$ and $A_2$. 
\begin{align}
f(W_{l_{k+1}})_{A_1}&=f\Big(\sum_{i=1}^{k-1} S_{i} + S_{k} + S_{{k+1}}\Big)\nonumber\\
&=f\Big(\sum_{i=1}^{k-1} S_{i} + S_{{k+1}} + S_{{k}}\Big)\nonumber 
\end{align}
\begin{align}
&=f(W_{l_{k+1}})_{A_2}. \label{batch:2}
\end{align}
Now, $f(W_{l_{k}})_{A_1}=f\Big(\sum\limits_{i=1}^{k-1} S_{i} + S_{k}\Big)$. Since $f$ is non-increasing in its argument, we have
\begin{align}
    f(W_{l_{k}})_{A_1} \geqslant f\Big(\sum\limits_{i=1}^{k-1} S_{i} + S_{{k+1}}\Big)=f(W_{l_{k}})_{A_2}.  \label{batch:3}
\end{align}
From (\ref{batch:1}), (\ref{batch:2}) and (\ref{batch:3}), we have $R_{A_1} \geqslant R_{A_2}$. 
\end{proof}
We observe that an arbitrary order of service is a permutation of the servicing order $A_2$ and that lemma \ref{proposition:Batch} can be extended to all such orders of service in place of $A_2$. More generally, lemma \ref{proposition:Batch} states that any work-conserving discipline that serves the jobs in increasing order of their sizes yields higher rewards. Examples of such service discipline include SRPT and preemptive shortest job first (PSJF) also. 

\section{Stochastic Arrivals}\label{Section:Sequential}
We now focus on the scenario with stochastic job arrivals. The goal here is quite different from that for batch arrivals. We cannot extend the results in section \ref{batch_arrivals} to this scenario as lemma \ref{proposition:Batch} does not hold here. Furthermore, the well-known service disciplines perform differently depending on the job sizes and the arrival rates. For instance, consider that the jobs of same size, $j$, arrive with interarrival times $\{Y_i, i\in \mathbb{N}\}$, where 
\begin{equation*}
    Y_i=
    \begin{cases}
    j_1=j+1-\delta \quad \text{w.p} \; \frac{1}{2}\\
    j_2=j+1+\delta \quad \text{otherwise}
    \end{cases}
\end{equation*}
for any $\delta>0$. Note that the system is stable with $\{Y_i, i\in \mathbb{N}\}$. Let $f(W_i)=\exp{(-\kappa W_i)}\; \forall i$ for some $\kappa>0$. 

\begin{figure}[hbpt!]
    \centering    \includegraphics[scale=0.58]{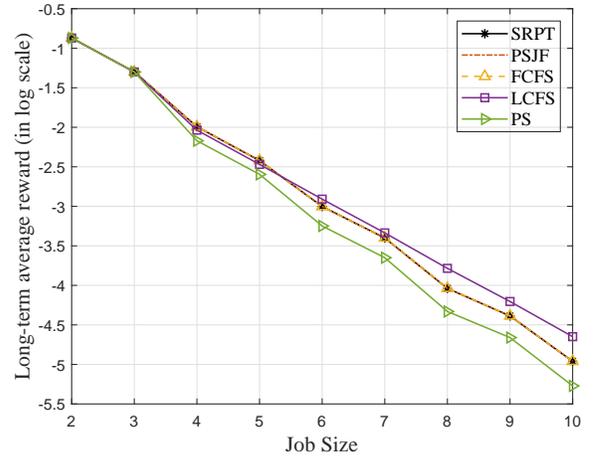}
    \caption{Job size vs. long-term average reward for $\kappa=1$}
    \label{fig:toy_illustration}
\end{figure}

For this arrival sequence with $\delta=\lfloor \frac{j}{2} \rfloor$ and $\kappa=1$, figure \ref{fig:toy_illustration} shows the performance of well-known disciplines: SRPT, PSJF, FCFS, last come first serve (LCFS), and processor sharing (PS). Since the jobs are of same size $j$, some disciplines perform the same. However, this is not the case for all arrival sequences. For $j=4$, FCFS, SRPT, and PSJF yield higher long-term average rewards, whereas LCFS dominates for $j\ge 6$. It is therefore evident from figure \ref{fig:toy_illustration} that even for a fixed $\kappa$, the performance of the aforementioned disciplines varies according to the job sizes. Next, we propose a new service discipline named shortest predicted sojourn time (SPST), which performs better than FCFS, LCFS, PSJF, SRPT, and PS in simulation. We also provide an analytical comparison with FCFS. 

\subsection{Shortest predicted sojourn time (SPST)}
The server of a work conserving queue cycles between idle and busy periods, i.e., the periods when the queue is empty and when it is not, respectively. On a given sample path of the arrival process and a given realization of the job size sequence, the positions and duration of the busy and idle periods are the same for all work conserving policies. Moreover, for an arrival process with i.i.d inter-arrival times, the beginning of a busy period is a renewal (or regenerative) epoch. Thus, by the renewal reward theorem \cite{Gallager2013}, for maximizing the long-term average reward, it is enough to maximize the average total reward in a renewal cycle.

For a fast decaying $f(\cdot)$, the total reward in a renewal cycle is dominated by the jobs with the shortest sojourn time. Thus, the two main factors that ensure high total reward in a cycle are the minimum sojourn time across all jobs in that cycle and the number of jobs whose sojourn time is equal to or close to that. 

As the future arrivals and job sizes are not known while making the service decision, intuitively, the best one can do is to serve the job whose completion would result into the shortest sojourn time among the existing jobs. This may increase the sojourn times of other jobs. However, as they are not the dominating terms in the total reward, the overall reward would be high.

Based on the above insights, we design the following policy, which we call shortest predicted sojourn time (SPST).



\begin{definition}[Predicted sojourn time]
Predicted sojourn time of a job at index $i$ at time $t \ge 0$ under a service discipline $\pi$, denoted by $P_\pi^{(t,i)}$, is its sojourn time if it is chosen by the server at time $t$ and is run to completion without preemption.
\end{definition}

\begin{definition}[Shortest predicted sojourn time (SPST)]\label{definition_SPST}
Under SPST, at every time instant, the job in service is the one with the shortest predicted sojourn time. That is, at any time $t$, the index of the job in service is
\begin{align*}
    k=\argmin\limits_i P_{SPST}^{(t,i)}.
\end{align*}
In case of a tie, the job with the least arrival time is prioritized.
\end{definition}

\subsection{Performance of SPST and other disciplines}\label{simulation_results}

In this subsection, the performance of SPST is compared with that of other well-known service disciplines. The long-term average rewards are plotted on a log scale for better visualization. We consider that the reward associated with each job is of the form $f(W)=\exp{(-\kappa W)}$ for some $\kappa>0$. Figures \ref{fig:toy_illustration_SPST} and \ref{fig:interarrival} depict the performance of disciplines when jobs of same size arrive with interarrival times $\{Y_i: i\ge 0\}$. We consider $\delta=\lfloor \frac{j}{2} \rfloor$ for the simulations. It is noted that SPST performs better than the existing disciplines for all job sizes for $\kappa=1$.  

\begin{figure}[hbpt!]
    \centering
    \includegraphics[scale=0.58]{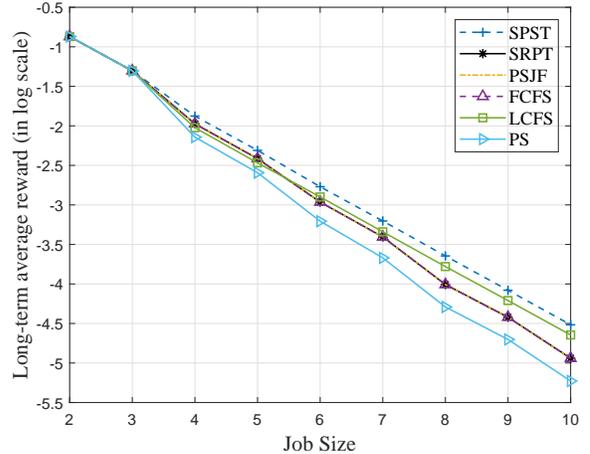}
    \caption{Job size vs. long-term average reward for $\kappa=1$. An illustration of the performance of SPST discipline.}
    \label{fig:toy_illustration_SPST}
\end{figure}

\begin{figure*}[hbpt!]
\centering
\captionsetup{justification=centering}
\subfloat[Job size vs. long-term average reward for $\kappa=1$\label{Bern_inter1}]{
\includegraphics[scale=0.58]{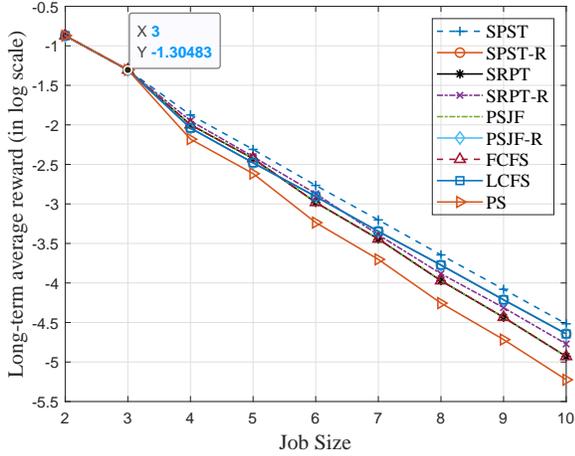}
}
\subfloat[$\kappa$ vs. long-term average reward for $j=4$. \label{Bern_inter2}]{
\includegraphics[scale=0.58]{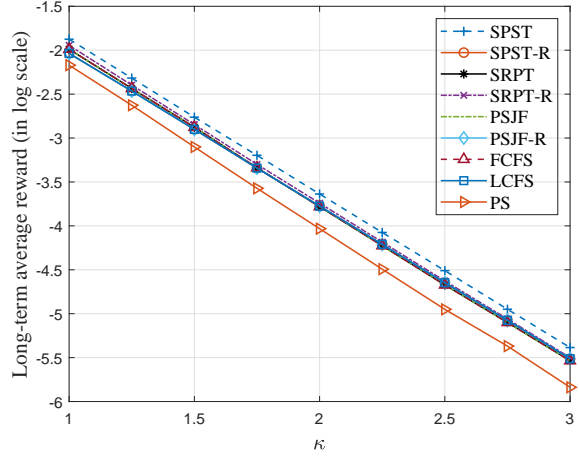}
}\\

\caption{An illustration for the case of job arrivals with $\{Y_i, i\in \mathbb{N}\}$ for $\delta=\lfloor \frac{j}{2} \rfloor$.}%
\label{fig:interarrival}
\end{figure*}

\begin{figure*}[hbpt!]
\centering
\captionsetup{justification=centering}
\subfloat[Job size vs. long-term average reward for $\kappa=1$\label{Bernoulli1}]{
\includegraphics[scale=0.58]{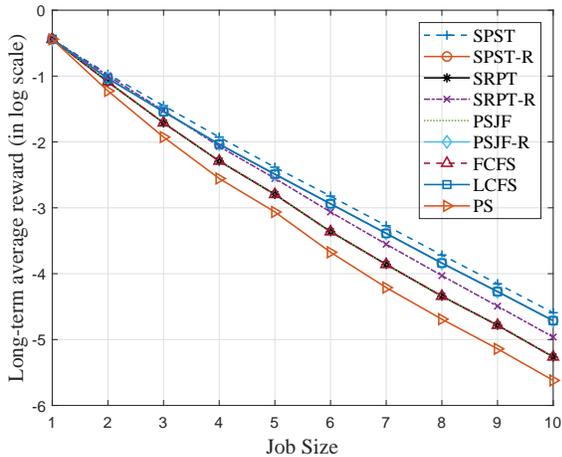}
}
\subfloat[$\kappa$ vs. long-term average reward for $j=4$. \label{Bernoulli2}]{
\includegraphics[scale=0.58]{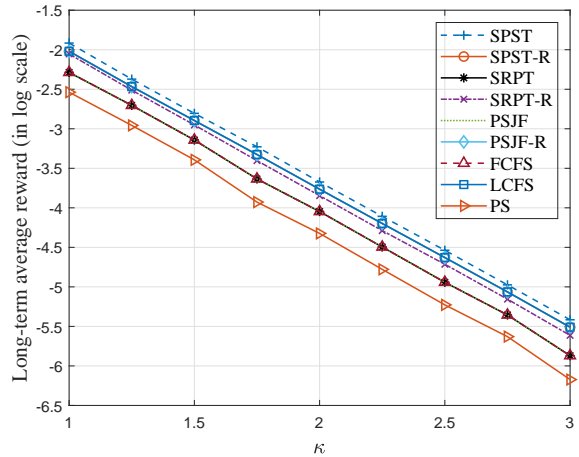}
}\\
\caption{An illustration for the case of Bernoulli arrivals with probability of arrival $\frac{1}{j+1}$.}%
\label{fig:Bernoulli}
\end{figure*}

By convention, the job with the least arrival time is prioritized for service in case of a tie under any service discipline. However, in our reward-based queue setup with $f(W)=\exp{(-\kappa W)}$, the tie-breaking criterion has to be suitably chosen to exploit the contribution of smaller jobs to the accumulated reward of the system. So, we also simulate the disciplines with a tie-breaker that prioritizes the most recent job for service. The suffix \emph{-R} represents the discipline with this tie-breaker. e.g., SPST-R. 

Figure \ref{Bern_inter1} depicts the performance of the disciplines along with their tie-breaking variant. With $j=2$, any busy period is $2$ irrespective of the service discipline, and hence, their long-term average rewards are the same. In addition, for any $i$, $Y_i$ is either 2 or 4 with equal probability, which is why the long-term average reward is 0.135. The same applies for $j=3$, in which case the long-term average reward is 0.0497. It is noted that SPST still yields rewards higher than that of any of its contenders for $\kappa=1$. In particular, for $\kappa \ge 1$, SPST is a clear winner for all sample paths regardless of $j$ as shown in \ref{Bern_inter2}. 

Figures \ref{Bernoulli1} and \ref{Bernoulli2} illustrate a more general case of Bernoulli arrivals with jobs of fixed size $j$. To ensure stability of the queue, we take arrival rate to be $\frac{1}{j+1}$. It can be seen that, even in this case, the performance of SPST is clearly better than the other policies for $\kappa \ge 1$. 

Since the jobs are of the same size in either case of arrival sequences, some disciplines perform equally. As seen in figures \ref{fig:interarrival} and \ref{fig:Bernoulli},
SPST-R, PSJF-R, and LCFS show similar performance. Likewise, the performances of SRPT, PSJF, and FCFS are similar. In addition, it is evident that the performance of PS is worse than that of SPST for $\kappa\ge 1$. This could mainly be due to the time-sharing of jobs under PS. It is also observed that SRPT-R and PSJF-R show better performance when compared to their respective conventional variants. However, under SPST, when same-sized jobs arrive in a sequence, there can never be a tie between two jobs waiting in the queue based on their predicted sojourn times. Only the job in service and a job in the queue are tied on this basis, in which case priority to the job in service yields better rewards. On the other hand, for more realistic arrival models with different job sizes, the server can choose a tie-breaker under SPST depending on the secondary performance measure such as expected slowdown\cite[Chapter~28]{Balter2013}.

Although simulations suggest that SPST is better and may even be an optimal policy for all arrival sequences, proving such guarantees are extremely hard. In the next section, we analytically prove that SPST performs better than FCFS. It will be evident that even this comparison is quite challenging due to intricate combinatorial structures.

\subsection{Analytical guarantee}\label{analytical_results}
For theoretical analysis, we consider a queuing system in which jobs of same size, $j$, arrive with interarrival times $\{Y_i: i\ge 0\}$. Recall
\begin{equation*}
    Y_i=
    \begin{cases}
    j_1=j+1-\delta \quad \text{w.p} \; \frac{1}{2}\\
    j_2=j+1+\delta \quad \text{otherwise}
    \end{cases}
\end{equation*}
We consider that the reward associated with each job is of the form $f(W)=\exp{(-\kappa W)}$ for $\kappa>0$. Under a stationary service discipline $\pi$, let $\bar{r}_{\pi}:=\lim_{n \to \infty} \frac{1}{n}\sum_{i=1}^n f(W_i)_\pi$, i.e., the long-term average reward. The following proposition is the main result of this section. 

\begin{Proposition}
For the defined queuing system with $\delta \le \frac{j}{2}$ and $f(W)=\exp{(-\kappa W})$, $\bar{r}_{SPST} \ge \bar{r}_{FCFS}$ for all $\kappa \ge \log_e 2$.\label{prop2}
\end{Proposition}
Proposition \ref{prop2} is a direct consequence of theorem \ref{theorem1}. A better understanding of this relationship requires the following definitions. 

\begin{definition}[Busy period] The time from when the server is busy until it becomes idle. \label{BusyPeriod}
\end{definition}
\begin{definition}[Busy period length]\label{BusyPeriodLength}
The number of jobs in a busy period is called its length, $L$.
\end{definition}
\begin{definition}[Idle period] The time from when the server is idle until it becomes busy. \label{IdlePeriod}
\end{definition}

Let $R_\pi$ denote the accumulated reward in an arbitrary busy period under service discipline $\pi$. That is, for a busy period of length $n$, $R_{\pi}=\sum\limits_{i=1}^n f(W_i)_\pi$. We denote the number of arrivals till $t$ by $\mathcal{A}(t)$. Then, $\bar{r}_{\pi}=\lim \limits_{t\xrightarrow{}\infty}\frac{1}{\mathcal{A}(t)}\sum \limits_{j=1}^{\mathcal{A}(t)} f(W_j)_{\pi}$. 

\begin{theorem}
For the defined queuing system with $\delta \le \frac{j}{2}$ and $f(W)=\exp{(-\kappa W})$, $R_{SPST}\ge R_{FCFS}$ for $\kappa \ge \log_e 2$. \label{theorem1}
\end{theorem}

By renewal reward theorem, we have
\begin{align*}
    \bar{r}_{\pi}&=\frac{\mathbb{E}[R_\pi]}{\lambda\mathbb{E}\text{[busy period + idle period]}}
\end{align*}
where $\lambda=\frac{1}{j+1}$ denotes the arrival rate of the jobs. We note that $\mathbb{E}\text{[busy period + idle period]}$ is the same for all work conserving disciplines. Therefore, by theorem \ref{theorem1}, we have
$\bar{r}_{SPST} \ge \bar{r}_{FCFS}$ for all $\kappa \ge \log_e 2$.

\section{Proof of Theorem \ref{theorem1}}\label{sequential:theorem}

Recall definitions \ref{BusyPeriod} and \ref{BusyPeriodLength}. The following are the observations with respect to a busy period for the case of job arrivals with  $\{Y_i, i\in \mathbb{N}\}$ defined earlier.
\begin{enumerate}[(i)]
    \item If the first inter-arrival time, $Y_1$, is $j_2$, then the busy period is $j$. In this case, any work conserving discipline yields the same reward, $\exp{(-\kappa j)}, \kappa>0$. \label{obs1}
    \item If $Y_1=j_1$, then $L>1$ for $\delta>1$. \label{obs2}
    \item A busy period has ended if $ (k_1+k_2)j \le k_1j_1+k_2j_2$ for $k_1, k_2 \ge 0$. This is because, in any work conserving discipline, the total work in a busy period cannot be greater than the busy period itself. \label{condition_busy}
\end{enumerate}

\begin{remark}\label{remark1}
Following observation \ref{condition_busy}, before a busy period ends, $k_2$ cannot be larger than $k_1$. However $k_1 \le k_2$ is only a sufficient condition for a busy period to end.
\end{remark}

We use the following lemmas to prove theorem \ref{theorem1}.

\begin{definition}[Priority job]
A job of size $j$ is called a priority job under any discipline if its sojourn time is $j$. In other words, a priority job neither waits nor is preempted until it is run to completion. 
\end{definition}

\begin{lemma}\label{lemma1}
Under SPST, there are at least $\lceil \frac{n}{2} \rceil$ priority jobs for $\delta \le \frac{j}{2}$.
\end{lemma} 

The following definitions are instrumental to understanding lemma \ref{lemma1} and the subsequent lemmas. Proof of the lemmas are given in the Appendix.

\begin{definition}[Block A]
The consecutive jobs that follow the interarrival time $j_1$ form a block A. 
\end{definition}

\begin{definition}[Block B]
The consecutive jobs that follow the interarrival time $j_2$ form a block B. 
\end{definition}

Let $n_A$ and $n_B$ denote the number of blocks A and B in the busy period respectively. We consider that $A_k$ denotes the $k^{th}$ block A, $n(A_k)$ denotes the number of jobs in the $k^{th}$ block A and $A_k^i$ denotes the $i^{th}$ job in the $k^{th}$ block A, with the similar interpretation for block B. Let $n_{j_1}$ and $n_{j_2}$ represent the total number of jobs in blocks A and B respectively. i.e., $\sum \limits_{i=1}^{n_A}n(A_i)=n_{j_1}$ and $\sum \limits_{i=1}^{n_B}n(B_i)=n_{j_2}$. Let $n_{\pi}^P$ represent the number of priority jobs in the busy period under a service discipline $\pi$.

\begin{lemma}
For any $A_i$ with odd $n(A_i)$, there exists a job whose sojourn time under SPST is $j+\delta-1$. 
\label{lemma2}
\end{lemma}

\begin{lemma}
If, in a busy period, $n$ is even and $n(A_i)$ is even for every $i$, then $n_{SPST}^P \ge \frac{n}{2}+1$.  
\label{lemma3}
\end{lemma}

\begin{lemma}
\label{lemma4}
Under FCFS,
~\newline
    \begin{minipage}[t]{\linewidth}
    \begin{enumerate}[a)]
    \item There is only one priority job. \label{claimFCFSa}
    \item All other jobs have $W\ge j+1$. \label{claimFCFSb}
    \item For $n\ge 3$, at least one job has $W\ge j+\delta$.
    \end{enumerate}
\end{minipage}
\end{lemma}

\begin{proof}[Proof of theorem \ref{theorem1}]
Let T denote the busy period of length $n$.
For $n<3$, $R_{SPST}=R_{FCFS}$.
For $n\ge 3$, from lemma \ref{lemma4},
\begin{align}
    R_{FCFS} &\le \exp{(-\kappa j)}+\exp{(-\kappa(j+\delta))}\nonumber\\
    &\quad+(n-2)\exp{(-\kappa(j+1))}\label{FCFS1} \\
    &\le \exp{(-\kappa j)}+(n-1)\exp{(-\kappa(j+1))}.\label{FCFS2}
\end{align}
If $n$ is odd, from lemma \ref{lemma1},
\begin{align}
    R_{SPST} &\ge \Big \lceil \frac{n}{2} \Big \rceil \exp{(-\kappa j)}+(n-\Big \lceil \frac{n}{2} \Big \rceil) \exp{(-\kappa T)}. \label{SPST1}
\end{align}
Using (\ref{FCFS2}) and (\ref{SPST1}), 
\begin{align}
    R_{SPST}&-R_{FCFS} \nonumber\\&\ge (\Big \lceil \frac{n}{2} \Big \rceil-1) \exp{(-\kappa j)}+(n-\Big \lceil \frac{n}{2} \Big \rceil) \exp{(-\kappa T)}\nonumber\\&\quad-(n-1)\exp{(-\kappa(j+1))}\nonumber\\
    &\ge (\Big \lceil \frac{n}{2} \Big \rceil-1) \exp{(-\kappa j)}-(n-1)\exp{(-\kappa(j+1))}\nonumber\\
    &\ge \Big \lfloor \frac{n}{2} \Big \rfloor \exp{(-\kappa j)} (1-2\exp{(-\kappa)}).\label{kappa1}
\end{align}
If $n$ is even, following lemmas \ref{lemma2} and \ref{lemma3},
\begin{align}
    R_{SPST} &\ge \frac{n}{2} \exp{(-\kappa j)}+\exp{(-\kappa (j+\delta-1))}\nonumber\\
    &\quad+(n-\frac{n}{2}-1) \exp{(-\kappa T)}. \label{SPST2}
\end{align}
Using (\ref{FCFS1}) and (\ref{SPST2}), 
\begin{align}
    R&_{SPST}-R_{FCFS}\nonumber\\&\ge \Big [\frac{n}{2}-1 \Big ] \exp{(-\kappa j)}+\exp{(-\kappa (j+\delta))}(\exp{(\kappa)}-1)\nonumber\\&\quad+(n-\frac{n}{2}-1) \exp{(-\kappa T)}-(n-2)\exp{(-\kappa(j+1))}\nonumber\\
    &\ge \Big [\frac{n}{2}-1 \Big ] \exp{(-\kappa j)}-(n-2)\exp{(-\kappa(j+1))}\nonumber\\
    &\ge \Big [\frac{n}{2}-1 \Big ]\exp{(-\kappa j)} (1-2\exp{(-\kappa)})\label{kappa2}.
\end{align}
From equations (\ref{kappa1}) and (\ref{kappa2}), $R_{SPST}\ge R_{FCFS}$ for $\kappa \ge \log_e 2$.
\end{proof}

Proposition \ref{prop2} states that the long-term average reward under SPST is more than that under FCFS for $\kappa\ge \log_e 2$. However it is also clear from (\ref{FCFS1}), (\ref{SPST1}), and (\ref{SPST2}) that  $\log_e 2$ is not a sharp threshold and obtaining a tight lower bound on $\kappa$ is far from simple. 

\section{Conclusion}
In this paper, we studied the problem of maximizing
the average nonlinear functions of sojourn times in work
conserving single server queuing systems and characterized
the performance of some well-known service disciplines. We
argued that identifying a single service discipline that
outperforms other disciplines for all arrival distributions and
job sizes appears to be a highly nontrivial problem. Indeed, an
optimal policy could depend on the specific functional form
of the nonlinear reward function. We also introduced a service
discipline, shortest predicted sojourn time (SPST), and pro-
vided analytical guarantees under specific settings. Numerical
experiments suggest that SPST performs well across multiple
settings, although it may not be optimal for all job sizes,
arrival distributions, and reward functions. As such, the general
problem setting remains largely open for further analytical
investigations.

\section*{Acknowledgement}
The first author's work was supported by the Prime Minister's Research Fellows (PMRF) scheme. The work of AC was supported in part by the Department of Science and Technology, Government of India under Grant SERB/SRG/2019/001809 and Grant INSPIRE/04/2016/001171. KJ acknowledges the Metro Area Quantum Access Network (MAQAN) project, supported by the Ministry of Electronics and Information Technology, India vide sanction number 13(33)/2020-CC\&BT. 

\section*{Appendix}\label{Appendix}
Recall the definitions and notations discussed in section \ref{Section:Sequential}. The following claim is required for the construction of the proof of~\cref{lemma1,lemma2,lemma3}.

\begin{Claim} \label{claimSPST} Under SPST discipline,~\newline
    \begin{minipage}[t]{\linewidth}
    \begin{enumerate}[a)]
    \item For a busy period of length $n$, $n_{j_1}+n_{j_2}=n-1$. \label{claim1.1}
    \item A busy period with $L>1$ always starts with block A. Also, every block B is preceded by a block A. That is, $n_A-n_B\in \{0,1\}$.\label{claim1.2} 
    \item For $\delta \le \frac{j}{2}$, the jobs in the even index of block A are priority jobs under SPST.\label{claim1.5}
    \item Every job of block B is a priority job under SPST. \label{claim1.6}
    
    \end{enumerate}
\end{minipage}
\end{Claim}
\begin{proof}[Proof of claim \ref{claimSPST}]~\newline
    \begin{minipage}[t]{\linewidth}
    \begin{enumerate}[a)]
    \item The first job in a busy period does not constitute either of the blocks.
    
    \item Follows observation \ref{obs2} and the construction of the blocks. 
    
    
    
    \item Follows the construction of the blocks and for $\delta \le \frac{j}{2}$, $Y_{k-1}+Y_{k}>j$ for any $k>1$. 
    
    \item Follows the construction of blocks B.
    
    \end{enumerate}
    \end{minipage}
\end{proof}
\subsection{Proof of lemma \ref{lemma1}}\label{section:lemma1}
\begin{proof}[\unskip\nopunct]

\noindent Following claim \ref{claimSPST}, $n_{SPST}^P=\sum \limits_{i=1}^{n_A} \lfloor \frac{n(A_i)}{2} \rfloor + n_{j_2}+1$. 

\begin{case}[$n_A=n_B$]\label{n_a=n_b}
\normalfont
\begin{align*}
    n_{SPST}^P&\ge \sum \limits_{i=1}^{n_A} \frac{n(A_i)-1}{2}+n_{j_2}+1\\
              &\ge \frac{n_{j_1}-n_A}{2}+n_{j_2}+1\\
              &\ge \frac{n_{j_1}-n_{j_2}}{2}+n_{j_2}+1 & (\because n_B\le n_{j_2})\\
              &\ge \frac{n-1}{2}+1 & (\text{from claim}\; \ref{claimSPST})
\end{align*}
which gives $n_{SPST}^P\ge \Big \lceil \frac{n}{2} \Big \rceil$.

\end{case}

\begin{case}[$n_A=n_B+1$] \label{n_a=n_b+1}~\newline
\normalfont
When $n$ is even,
\begin{align*}
    n_{SPST}^P&\ge \sum \limits_{i=1}^{n_A} \frac{n(A_i)-1}{2}+n_{j_2}+1\\
              &\ge \frac{n_{j_1}-n_A}{2}+n_{j_2}+1\\
              &\ge \frac{n_{j_1}-(n_{j_2}+1)}{2}+n_{j_2}+1 & (\because n_B\le n_{j_2})\\
              &\ge \frac{n-2}{2}+1 & (\text{from claim}\; \ref{claimSPST})\\
              &\ge \frac{n}{2}.
\end{align*}
When $n$ is odd, there are four possible subcases as follows.
\begin{subcase}[$n_A$ is odd, $n_{j_1}$ is odd]\label{subcase1}
\normalfont
It is noted that $n_B$ is even and $n_{j_2}$ is odd (from claim \ref{claimSPST}). This implies that at least one block B has even number of jobs. Therefore, $n_B \le n_{j_2}-1$.
\begin{align*}
    n_{SPST}^P&\ge \sum \limits_{i=1}^{n_A} \frac{n(A_i)-1}{2}+n_{j_2}+1\\
              &\ge \frac{n_{j_1}-n_A}{2}+n_{j_2}+1\\
              &\ge \frac{n_{j_1}-n_{j_2}}{2}+n_{j_2}+1 & (\because n_B\le n_{j_2}-1)\\
              &\ge \frac{n-1}{2}+1 & (\text{from claim}\; \ref{claimSPST})\\
              &\ge \Big \lceil \frac{n}{2} \Big \rceil.
\end{align*}
\end{subcase}
\begin{subcase}[$n_A$ is even, $n_{j_1}$ is odd]\label{subcase2}
\normalfont
Here $n_B$ is odd and $n_{j_2}$ is odd. This implies that at least one block A has even number of jobs, Say, one such block is $A_{k'}$. 
\begin{align*}
    n_{SPST}^P&\ge \frac{n(A_{k'})}{2}+\sum \limits_{i=1}^{n_B} \frac{n(A_i)-1}{2}+n_{j_2}+1\\
              &\ge \frac{n_{j_1}-n_B}{2}+n_{j_2}+1\\
              &\ge \frac{n_{j_1}-n_{j_2}}{2}+n_{j_2}+1 & (\because n_B\le n_{j_2})\\
              &\ge \frac{n-1}{2}+1 & (\text{from claim}\; \ref{claimSPST})\\
              &\ge \Big \lceil \frac{n}{2} \Big \rceil.
\end{align*}
\end{subcase}
\begin{subcase}[$n_A$ is even, $n_{j_1}$ is even]
\normalfont
In this case, there is at least one block B that has even number of jobs. The lower bound for $n_{SPST}^P$ follows subcase \ref{subcase1} giving $n_{SPST}^P\ge \Big \lceil \frac{n}{2} \Big \rceil$.

\end{subcase}

\begin{subcase}[$n_A$ is odd, $n_{j_1}$ is even]\label{subcase4}
\normalfont
There is at least one block A that has even number of jobs and hence, the lower bound for $n_{SPST}^P$ in this case follows subcase \ref{subcase2} giving $n_{SPST}^P\ge \Big \lceil \frac{n}{2} \Big \rceil$.
\end{subcase}
\end{case}
From cases \ref{n_a=n_b} and \ref{n_a=n_b+1}, it is proved that there are at least $\lceil \frac{n}{2} \rceil$ priority jobs under SPST discipline for $\delta \le \frac{j}{2}$ .
\end{proof}

\subsection{Proof of lemma \ref{lemma2}}\label{section:lemma2}
\begin{proof}[\unskip\nopunct]

For any $i$, if $n(A_i)$ is odd, following claim \ref{claimSPST}, $A_i^{n(A_i)-1}$ is a priority job. 
Although block $A_i$ might be followed by block $B_i$, since $j_1+j_2>2j$, the sojourn time of $A_i^{n(A_i)}$ is governed only by its preceding priority job, $A_i^{n(A_i)-1}$, and hence, its waiting time in the queue is $\delta-1$. Therefore, the sojourn time of $A_i^{n(A_i)}$ is $j+\delta-1$.
\end{proof}

\subsection{Proof of lemma \ref{lemma3}}\label{section:lemma3}
\begin{proof}[\unskip\nopunct]

For a busy period, assume that $n$ is even and there is at least one block $A_j$ with even $n(A_j)$. By the analysis of cases \ref{n_a=n_b} and \ref{n_a=n_b+1} as in lemma \ref{lemma1}, we get $n^P_{SPST}\ge \frac{n}{2}+1$. Let $m_{e}$ denote the number of blocks $A_k$ with even $n(A_k)$. Extending the above-mentioned argument to all $m_e \ge 2$ such blocks, there are at least $\frac{n}{2}+\frac{m_e}{2}$ priority jobs in the busy period. In other words, except for one even block A, the presence of all other even blocks A improves the bound by $\frac{1}{2}$.
\end{proof}

\subsection{Proof of lemma \ref{lemma4}}\label{section:lemma4}
\begin{proof}[\unskip\nopunct]
~\newline
    \begin{minipage}[t]{\linewidth}
    \begin{enumerate}[a)]
    \item Under FCFS, the first job in a busy period is not preempted by any of the subsequent arrivals.
    \item Follows lemma \ref{lemma4}\ref{claimFCFSa}.
    \item Let $W_k^{FCFS}$ denote the sojourn time of the $k^{th}$ job arrival in the busy period under FCFS.  
    From lemma \ref{lemma4}\ref{claimFCFSa}, we know that $W_{1}^{FCFS}=j$ and for $k\ge 2$, $W_{k}^{FCFS}=W_{k-1}^{FCFS}-Y_{k-1}+j\ge j+1$. Owing to claim \ref{remark1}, for $n\ge 3$, there exists at least one job that follows the interarrival time $j_1$.
    Therefore, for such jobs, $k\ge 3$, 
    \begin{align*}
        W_{k}^{FCFS}&\ge 2j-Y_{k-1}+1\ge j+\delta
    \end{align*}
    \end{enumerate}
\end{minipage}
\end{proof}
\bibliographystyle{IEEEtran}
\bibliography{references}

\begin{thebibliography}{10}
\providecommand{\url}[1]{#1}
\csname url@samestyle\endcsname
\providecommand{\newblock}{\relax}
\providecommand{\bibinfo}[2]{#2}
\providecommand{\BIBentrySTDinterwordspacing}{\spaceskip=0pt\relax}
\providecommand{\BIBentryALTinterwordstretchfactor}{4}
\providecommand{\BIBentryALTinterwordspacing}{\spaceskip=\fontdimen2\font plus
\BIBentryALTinterwordstretchfactor\fontdimen3\font minus
  \fontdimen4\font\relax}
\providecommand{\BIBforeignlanguage}[2]{{%
\expandafter\ifx\csname l@#1\endcsname\relax
\typeout{** WARNING: IEEEtran.bst: No hyphenation pattern has been}%
\typeout{** loaded for the language `#1'. Using the pattern for}%
\typeout{** the default language instead.}%
\else
\language=\csname l@#1\endcsname
\fi
#2}}
\providecommand{\BIBdecl}{\relax}
\BIBdecl

\bibitem{Potts2009}
C.~N. Potts and V.~A. Strusevich, ``Fifty years of scheduling: A survey of
  milestones,'' \emph{The Journal of the Operational Research Society},
  vol.~60, pp. s41--s68, 2009.

\bibitem{Miller1966}
L.~E. Schrage and L.~W. Miller, ``The queue {M/G/1} with the shortest remaining
  processing time discipline,'' \emph{Operations Research}, vol.~14, no.~4, pp.
  670--684, 1966.

\bibitem{Nielsen2010}
M.~A. Nielsen and I.~L. Chuang, \emph{Quantum Computation and Quantum
  Information: 10th Anniversary Edition}.\hskip 1em plus 0.5em minus
  0.4em\relax Cambridge University Press, 2010, ch.~8.

\bibitem{krishna2019}
K.~Jagannathan, A.~Chatterjee, and P.~Mandayam, ``Qubits through queues: The
  capacity of channels with waiting time dependent errors,'' in \emph{2019
  National Conference on Communications (NCC)}, 2019, pp. 1--6.

\bibitem{Kurose2012}
J.~F. Kurose and K.~W. Ross, \emph{Computer Networking: A Top-Down Approach},
  6th~ed.\hskip 1em plus 0.5em minus 0.4em\relax Pearson, 2012, ch.~7.

\bibitem{Daugherty2010}
J.~R. Daugherty and G.~L. Brase, ``Taking time to be healthy: Predicting health
  behaviors with delay discounting and time perspective,'' \emph{Personality
  and Individual Differences}, vol.~48, no.~2, pp. 202--207, 2010.

\bibitem{Dewan1990}
S.~Dewan and H.~Mendelson, ``User delay costs and internal pricing for a
  service facility,'' \emph{Management Science}, vol.~36, no.~12, pp.
  1502--1517, 1990.

\bibitem{Schrage1968}
L.~Schrage, ``A proof of the optimality of the shortest remaining processing
  time discipline,'' \emph{Operations Research}, vol.~16, no.~3, pp. 687--690,
  1968.

\bibitem{Smith1978}
D.~R. Smith, ``Technical note - a new proof of the optimality of the shortest
  remaining processing time discipline,'' \emph{Oper. Res.}, vol.~26, pp.
  197--199, 1978.

\bibitem{Balter2013}
M.~Harchol-Balter, \emph{Performance Modeling and Design of Computer Systems:
  Queueing Theory in Action}, 1st~ed.\hskip 1em plus 0.5em minus 0.4em\relax
  Cambridge University Press, 2013.

\bibitem{Balter2001}
N.~Bansal and M.~Harchol-Balter, ``Analysis of {SRPT} scheduling: Investigating
  unfairness,'' in \emph{Proceedings of the 2001 ACM SIGMETRICS International
  Conference on Measurement and Modeling of Computer Systems}, ser. SIGMETRICS
  '01.\hskip 1em plus 0.5em minus 0.4em\relax New York, NY, USA: Association
  for Computing Machinery, 2001, p. 279–290.

\bibitem{Balter2000}
M.~Harchol-Balter, N.~Bansal, B.~Schroeder, and M.~Agrawal, ``Implementation of
  {SRPT} scheduling in web servers,'' 04 2001.

\bibitem{Balter2018}
I.~Grosof, Z.~Scully, and M.~Harchol-Balter, ``{SRPT} for multiserver
  systems,'' \emph{Performance Evaluation}, vol. 127-128, pp. 154--175, 2018.

\bibitem{Vaze2020}
R.~Vaze and J.~Nair, ``Multiple server {SRPT} with speed scaling is
  competitive,'' \emph{IEEE/ACM Transactions on Networking}, vol.~28, no.~4,
  pp. 1739--1751, 2020.

\bibitem{Barrer1957}
D.~Y. Barrer, ``Queuing with impatient customers and ordered service,''
  \emph{Operations Research}, vol.~5, no.~5, pp. 650--656, 1957.

\bibitem{Daley1965}
D.~J. Daley, ``General customer impatience in the queue {GI/G/1},''
  \emph{Journal of Applied Probability}, vol.~2, pp. 186--205, 1965.

\bibitem{Baccelli1984}
F.~Baccelli, P.~Boyer, and G.~Hebuterne, ``Single-server queues with impatient
  customers,'' \emph{Advances in Applied Probability}, vol.~16, no.~4, p.
  887–905, 1984.

\bibitem{Kargahi2004}
M.~Kargahi and A.~Movaghar, ``A method for performance analysis of
  earliest-deadline-first scheduling policy,'' in \emph{International
  Conference on Dependable Systems and Networks, 2004}, 2004, pp. 826--834.

\bibitem{Pascal2013}
P.~Moyal, ``On queues with impatience: Stability, and the optimality of
  earliest deadline first,'' \emph{Queueing Syst. Theory Appl.}, vol.~75, no.
  2–4, p. 211–242, nov 2013.

\bibitem{Ahmadi2021}
M.~Ahmadi, M.~Golkarifard, A.~Movaghar, and H.~Yousefi, ``Processor sharing
  queues with impatient customers and state-dependent rates,'' \emph{IEEE/ACM
  Transactions on Networking}, vol.~29, no.~6, pp. 2467--2477, 2021.

\bibitem{Balter2015}
K.~Gardner, S.~Borst, and M.~Harchol-Balter, ``Optimal scheduling for jobs with
  progressive deadlines,'' in \emph{2015 IEEE Conference on Computer
  Communications (INFOCOM)}, 2015, pp. 1113--1121.

\bibitem{Liu1973}
C.~L. Liu and J.~W. Layland, ``Scheduling algorithms for multiprogramming in a
  hard-real-time environment,'' \emph{J. ACM}, vol.~20, no.~1, p. 46–61, jan
  1973.

\bibitem{Towsley1988}
S.~S. Panwar, D.~Towsley, and J.~K. Wolf, ``Optimal scheduling policies for a
  class of queues with customer deadlines to the beginning of service,''
  \emph{J. ACM}, vol.~35, no.~4, p. 832–844, oct 1988.

\bibitem{Raviv2018}
L.-O. Raviv and A.~Leshem, ``Maximizing service reward for queues with
  deadlines,'' \emph{IEEE/ACM Transactions on Networking}, vol.~26, no.~5, pp.
  2296--2308, 2018.

\bibitem{Hyon2020}
E.~Hyon and A.~Jean-Marie, ``Optimal control of admission in service in a queue
  with impatience and setup costs,'' \emph{Performance Evaluation}, vol. 144,
  p. 102134, 2020.

\bibitem{Yu2018}
Z.~Yu, Y.~Xu, and L.~Tong, ``Deadline scheduling as restless bandits,''
  \emph{IEEE Transactions on Automatic Control}, vol.~PP, pp. 1--1, 02 2018.

\bibitem{Draper2005}
S.~Draper, M.~Trott, and G.~Wornell, ``A universal approach to queuing with
  distortion control,'' \emph{IEEE Transactions on Automatic Control}, vol.~50,
  no.~4, pp. 532--537, 2005.

\bibitem{Avhishek2017}
A.~Chatterjee, D.~Seo, and L.~R. Varshney, ``Capacity of systems with
  queue-length dependent service quality,'' \emph{IEEE Transactions on
  Information Theory}, vol.~63, no.~6, pp. 3950--3963, 2017.

\bibitem{chaudhry1983}
M.~Chaudhry and J.~Templeton, \emph{A First Course in Bulk Queues}, ser. A
  Wiley-interscience publication.\hskip 1em plus 0.5em minus 0.4em\relax Wiley,
  1983, ch.~2.

\bibitem{Gallager2013}
R.~G. Gallager, \emph{Stochastic Processes: Theory for Applications}.\hskip 1em
  plus 0.5em minus 0.4em\relax Cambridge University Press, 2013, ch.~5.

\end{thebibliography}

\end{document}